\documentclass[conference]{IEEEtran}
\IEEEoverridecommandlockouts
\usepackage{amsmath}
\usepackage[pdftex]{graphicx}
\graphicspath{{./Figures/}}
\ifCLASSOPTIONcompsoc
  \usepackage[caption=false,font=footnotesize,labelfont=sf,textfont=sf,subrefformat=parens,labelformat=parens]{subfig}
\else
  \usepackage[caption=false,font=footnotesize,subrefformat=parens,labelformat=parens]{subfig}
\fi

\usepackage{booktabs}

\usepackage{amsthm}

\newtheorem{lemma}{Lemma}

\usepackage{hyperref}
\usepackage{xurl}

\newcommand{\etal}{\textit{et al.}}
\newcommand{\Profit}[1]{\pi_{#1}}

\newcommand{\Participants}{p}
\newcommand{\nParticipants}{N_{\Participants}}
\newcommand{\EntryFee}{f}

\newcommand{\Prize}[1]{P_{#1}}

\newcommand{\totalPrizes}{P}
\newcommand{\CommissionFeeRatio}{r}
\newcommand{\Title}{Web3 Meets Behavioral Economics: An Example of Profitable Crypto Lottery Mechanism Design}

\def\BibTeX{{\rm B\kern-.05em{\sc i\kern-.025em b}\kern-.08em
    T\kern-.1667em\lower.7ex\hbox{E}\kern-.125emX}}
\begin{document}

\title{
    \Title\\
    \thanks{This work was partially supported by the Grant KAKENHI (No.18K18162) from MEXT/JSPS, Japan.}
}

\author{
    \IEEEauthorblockN{Kentaroh Toyoda}
    \IEEEauthorblockA{
        \textit{Institute of High Performance Computing (IHPC), Agency for Science, Technology and Research (A*STAR), Singapore} \\
        \textit{Keio University, Faculty of Science and Technology, Yokohama, Japan}\\
        0000-0002-6233-3121
    }
}

\maketitle

\begin{abstract}
We are often faced with the non-trivial task of designing incentive mechanisms in the era of Web3. As history has shown, many Web3 services failed mostly due to the lack of a rigorous incentive mechanism design based on token economics.
However, traditional mechanism design, where there is an assumption that the users of services strategically make decisions so that their expected profits are maximized, often does not capture their real behavior well as it ignores humans' psychological bias in making decisions under uncertainty.
In this paper, we propose an incentive mechanism design for crypto-enabled services using behavioral economics. Specifically, we take an example of a crypto lottery game in this work and incorporate a seminal work of cumulative prospect theory into its lottery game mechanism (or rule) design. We designed four mechanisms and compared them in terms of utility, a metric of how appealing a mechanism is to participants, and a game operator's expected profit. Our approach is generic and will be applicable to a wide range of crypto-based services where a decision has to be made under uncertainty.
\end{abstract}


\section{Introduction}
\label{sec:introduction}
We are often faced with the non-trivial task of designing incentive mechanisms in the era of Web3. The very first blockchain, Bitcoin, also incorporates an incentive mechanism into its core consensus algorithm dubbed proof-of-work, where miners compete with each other in its block generation process in return for cryptocurrency. A simple yet elegant incentive mechanism was designed by Nakamoto behind the success of Bitcoin~\cite{Nakamoto2008-nw}. Now that tokens are easily issued via smart contracts (e.g. via ERC-20 (Ethereum Request for Comments 20)), a large number of new crypto-based services emerged. However, as history has shown, many cryptocurrencies failed mostly due to the lack of a rigorous incentive mechanism design based on token economics~\cite{Toepffer2020-kv}.

Mechanism design, a field of microeconomics and game theory, helps us derive an optimal mechanism where desired objectives are achieved by incentives. Desired objectives here mean that a service, an application, or a system works as intended by its operators.
In the mechanism design, we start with modeling entities involved in the system such as operators and users. We then devise entities' profit or utility, which is often derived by subtracting costs from rewards. In the example of Bitcoin's mining, incentives are the rewards of newly minted Bitcoins and transaction fees involved in the block, and costs are electric costs consumed by mining. The traditional mechanism design assumes that entities strategically make decisions so that their expected profits are (mathematically) maximized. However, this often does not explain their real behavior well as it ignores humans' psychological bias in making decisions under uncertainty. For instance, gamblers play a lottery even though they know that its expectation is negative. Hence, we need an alternative approach to better capture participants' behavior in designing a mechanism.

In this paper, we propose an incentive mechanism design for crypto-enabled services using behavioral economics. Behavioral economics combines elements of economics and psychology to understand how and why people behave the way they do in the real world~\cite{Witynski_undated-xb}. Cumulative prospect theory (CPT) is a seminal work of behavioral economics that captures humans' bias in making decisions under risk~\cite{Tversky1992-vu}. The key idea of applying CPT is to ``transform'' traditional expectation functions (i.e. utility and probability) so that humans' bias is better explained, and we will revisit this with a detailed explanation in Section~\ref{sec:proposed_method}. We take an example of a crypto lottery game in this work and leverage CPT to design a profitable lottery game mechanism for a game operator. We design four mechanisms and compare them in terms of utility, a metric of how appealing a mechanism is for participants, and a game operator's expected profit. We rigorously test possible parameters and functions of CPT to identify the conditions where the game is appealing to participants and profitable for the game operator.

The contributions of our work are as follows.
\begin{itemize}
    \item To the best of our knowledge, we are the first to apply behavioral economics to mechanism design in the crypto-based service domain.
    \item Our approach is generic and will be applicable to a wide range of crypto-based services where a decision has to be made under uncertainty.
\end{itemize}

The rest of this paper is organized as follows:
Section~\ref{sec:problem_statement} describes the problem statement.
Section~\ref{sec:model} presents a model used in this research.
Section~\ref{sec:proposed_method} describes the proposed method.
Section~\ref{sec:performance_evaluation} shows numerical results and discusses novelties and open questions.
Section~\ref{sec:conclusions} concludes this paper.

\section{Problem Statement}
\label{sec:problem_statement}
There are few attempts to apply behavioral economics in the crypto and Web3 domains. As far as we know, the only work was done by Thoma, which investigated the risk and return properties of cryptocurrency trading and found that prospect theory well explains the attractiveness of cryptocurrency trading~\cite{Thoma2020-wl}. Hence, although it would be beneficial to incorporate token economics in designing profitable crypto-based services, none of the work has taken into account human behavioral perspectives to their incentive mechanisms.\footnote{Of course, the classical approach based on expected utility maximization would work fine when such a human bias is not involved.} We take a first step towards introducing a powerful ``toolset'' of behavioral economics, CPT, to design incentive mechanisms for crypto-based services. Specifically, we chose a crypto-based lottery game as its example as some of us are going to launch one.\footnote{There exists plenty of work that studies gambles from an economic perspective (e.g.~\cite{Tang2010-zw, Barberis2012-dr, Kesten2017-bh, Baker2020-nd}). It is challenging to cite them here properly. Yet, one of the closest works is done by Tang \etal\ that analyzed the pricing and design of three lotteries, namely a single prize, multiple equal prizes, and multiple weighted prizes, with CPT~\cite{Tang2010-zw}. The optimal pricing of each lottery was derived by solving the first- and second-order conditions of lottery buyers' value function, maximizing their utilities.} The operators of crypto-based lottery games are interested in how to maximize their profits, and it is quite important to predict the expected profits of their games. In this regard, it is vital to understand the behavior of participants. A naive approach is to design game rules and analyze the expected profits of operators and participants. Specifically, when a game rule is given, potential participants need to determine whether or not they should play it based on the expectation of returns. If the expectation of returns is positive, participants are expected to join the game. However, as it will be shown later, this naive approach would not explain their actual behavior well as their action is often biased by psychological factors~\cite{Kahneman1979-pc}. The objectives of our research are to suggest designing profitable crypto-based lotteries with behavioral economics and to show how to quantify its goodness with numerical analysis. 

\section{Model}
\label{sec:model}
\begin{table}[t]
    \centering
    \caption{Notation table.}
    \label{tab:notation_table}
    \begin{tabular}{p{.1\columnwidth}p{.8\columnwidth}}
    \toprule
    Symbol & Description \\
    \midrule
    $\nParticipants$ & Number of participants \\
    $\EntryFee$ & Entry fee \\
    $\CommissionFeeRatio$ & Ratio of the amount that an operator takes from collected entry fees \\
    $\totalPrizes$ & Total amount of prizes given to winners (i.e. $\nParticipants \cdot \EntryFee \cdot (1 - \CommissionFeeRatio$)) \\
    $\Prize{j}$ & Prize given to the $j$-th participant \\
    $\Profit{j}$ & Profit of participant who ranked in $j$-th (i.e. $\Prize{j} - \EntryFee$) \\
    $U(\cdot)$ & Utility function in the EUT \\
    $v(\cdot)$ & Value function in the CPT \\
    $w(\cdot)$ & Probability weighting function in the CPT \\
    $k$ & Ratio of participants that can receive prizes in the top-$k$ mechanisms \\
    \bottomrule
    \end{tabular}
\end{table}
We first model the entities and game design that we analyze in this work. \tablename~\ref{tab:notation_table} lists the notation used in this paper.

\subsection{Entities}
There are two entities, namely a game operator and participants. An operator first determines a lottery game rule. Given a rule, participants determine whether or not to join a game.

\subsection{Lottery Games}
We define our generic lottery game as follows. A participant needs to pay an entry fee $\EntryFee$ to join a game. We assume that participants are ordered by their ranks at the end of a game. The $j$-th ranker will be given a prize $\Prize{j}$, and $\Prize{1} \geq \Prize{2} \geq \dots \geq \Prize{\nParticipants} \geq 0$ where $\sum_{j = 1}^{\nParticipants}\Prize{j} = \totalPrizes$. The source of prizes is collected entry fees. An operator of the game takes $\CommissionFeeRatio$ of the total amount of entry fees, say $\CommissionFeeRatio = 10\%$, for their revenue. 

One such lottery game example is a top growth rate game where some top participants can gain prizes based on their portfolios' growth rate at the end of a game period.\footnote{We will test and deploy this particular game as a service.} Participants are to invest in cryptocurrencies during a game period and compete with each other in their return on investment.
We assume that no one can predict the future price of cryptocurrencies and that the rankings of participants' growth rates are determined by random sampling from a uniform distribution. As we only need to determine the ranking of participants' growth rates, when we let $\theta \in [0, 1]$ denote the type of a participant (or relative competitiveness among participants), the smaller $\theta$, the higher chance of being ranked in a higher position (and vice versa). As $\theta$ is assumed to follow a uniform distribution, its cumulative density function is simply denoted as $F(\theta) = \theta$. Note that $\theta$ is a priori unknown; it is revealed when rankings are determined, meaning that no participant knows his/her own $\theta$ at the time of joining a game. Hence, what participants can strategically determine is whether or not to join a game before it starts giving a game rule or mechanism, and it thus can be seen as a lottery game.

\subsection{Profits}
The profit of an operator and participants can be derived as follows. The operator's profit (revenue) can be obtained by multiplying the number of participants $\nParticipants$ by their entry fees $\EntryFee$ and how much he/she takes $\CommissionFeeRatio$. 
\begin{equation}
    \nParticipants \cdot \EntryFee \cdot \CommissionFeeRatio.
\end{equation}
$\EntryFee$ and $\CommissionFeeRatio$ can be determined by an operator while $\nParticipants$ cannot. Hence, an operator has to determine the mechanism that more participants will join (i.e. larger $\nParticipants$) when $\EntryFee$ and $\CommissionFeeRatio$ are fixed. 

On the other hand, the profit of a participant who ended up $j$-th position is determined by their rankings of growth profit.\footnote{As gamblers enjoy the game itself, a utility function may need to include positive factors of ``excitement'' and ``entertainment''~\cite{Stetzka2021-nx}. However, we simply assume that the positive factor of profit only comes from a game's prize.}
\begin{equation}
    \Profit{j} = \Prize{j} - \EntryFee. \label{eq:profit_participants}
\end{equation}
For an operator to predict if participants will join a game, a naive approach would be to calculate the expectation of their profit and check if it is positive. More specifically, expected utility theory (EUT) is often used to analyze the choice and behavior of players. In the EUT, they make decisions so that their expected utilities are maximized. A utility here means a value of satisfaction given an outcome. For instance, when a participant receives $\Prize{j}$, his/her utility is represented as a conversion of $\Prize{j}$ to $U(\Prize{j})$. A utility function $U(\cdot)$ is increasing, but its curvature is determined by a player's risk preference (e.g. curvature is concave when a player is risk-averse while it is linear ($U(\Profit{j}) = \Profit{j}$) when a player is risk-neutral). Hence, an expected utility is represented as follows.
\begin{equation}
    \sum_{j = 1}^{\nParticipants} p_j U(\Profit{\Participants}) = \frac{1}{\nParticipants}\sum_{j = 1}^{\nParticipants} U(\Profit{\Participants})
    \label{eq:EUT}
\end{equation}
Note that when a participant is assumed to be risk-neutral, the expected utility coincides with the expectation.

However, it is obviously unrealistic to assume that participants follow EUT as this implies that no one will join a game as proven below. 
\begin{lemma}
Individual Rationality: When (risk-neutral) participants are assumed to strategize their participation under the expected utility assumption, no participant would join the defined game regardless of its mechanism.
\end{lemma}
\begin{proof}
We first derive a condition that each participant joins the game. For a rational participant joins the game, $E[\Profit{\Participants}]$ must be larger than or equal to zero. However, $\theta$ is revealed only after a participant decides to join a game. Hence, regardless of the value of $\theta$, we simply find the expected utility of participants from Eq.~\ref{eq:profit_participants}. As the probability for a participant to be given each prize is $1 / \nParticipants$ for all $\Prize{j}$, the expectation is derived as follows.
\begin{align}
    E[\Profit{\Participants}] &= \frac{1}{\nParticipants} \sum_{j = 1}^{\nParticipants} \Profit{\Participants}, \nonumber \\
    &= \frac{1}{\nParticipants} \sum_{j = 1}^{\nParticipants} \Prize{j} - \EntryFee, \nonumber \\
    &= \frac{1}{\nParticipants} \nParticipants \cdot \EntryFee \cdot (1 - \CommissionFeeRatio) - \EntryFee, \nonumber \\
    &= - \CommissionFeeRatio \cdot \EntryFee. \label{eq:expected_profit_participants}
\end{align}
We used $\sum_{j = 1}^{\nParticipants} \Prize{j} = \nParticipants \cdot \EntryFee \cdot (1 - \CommissionFeeRatio)$. As $- \CommissionFeeRatio \cdot \EntryFee$ is negative, the expected utility of participants is always negative regardless of the mechanism.
\end{proof}

However, as research revealed (e.g.~\cite{Barberis2013-ar, Baker2020-nd}), this would not capture the actual behavior of participants (or gamblers), and their behavioral bias must be incorporated into game mechanism design and analysis. 

\section{Proposed Method}
\label{sec:proposed_method}
We propose a method for an operator to be able to design a lottery game mechanism that takes into account participants' behavioral bias. 
Again, the objectives of an operator and participants are both maximizing their profits. In this regard, an operator needs to determine a profitable mechanism, i.e. $\Prize{j}$ that gives an operator as well as participants more profits, and given a mechanism participants strategically determine whether or not to join a game under not EUT but CPT. Although it would be ideal to mathematically derive the most profitable structure of $\Prize{j}$ ($1 \leq j \leq \nParticipants$) from participants' utility under CPT via an appropriate method such as first- and second-order derivative, it is not easily tractable due to the complex structure of functions in CPT. Hence, in this paper, we propose a simple-yet-beneficial method to achieve the goal of the operator. Our method is composed of the following steps.
\begin{enumerate}
    \item Model entities' behavior.
    \item List mechanisms, i.e. a set of $\{\Prize{1}, \dots, \Prize{j}, \dots, \Prize{\nParticipants}\}$, $\EntryFee$ and $\CommissionFeeRatio$.
    \item Calculate participants' utilities with CPT based on the given mechanisms.
    \item Compare the mechanisms in terms of utility to choose the most profitable mechanism for an operator.
\end{enumerate}
As we have already explained the first step in the previous section, we detail the remaining three steps.

\subsection{Mechanisms}
\begin{figure}
    \centering
    \includegraphics[width=\columnwidth]{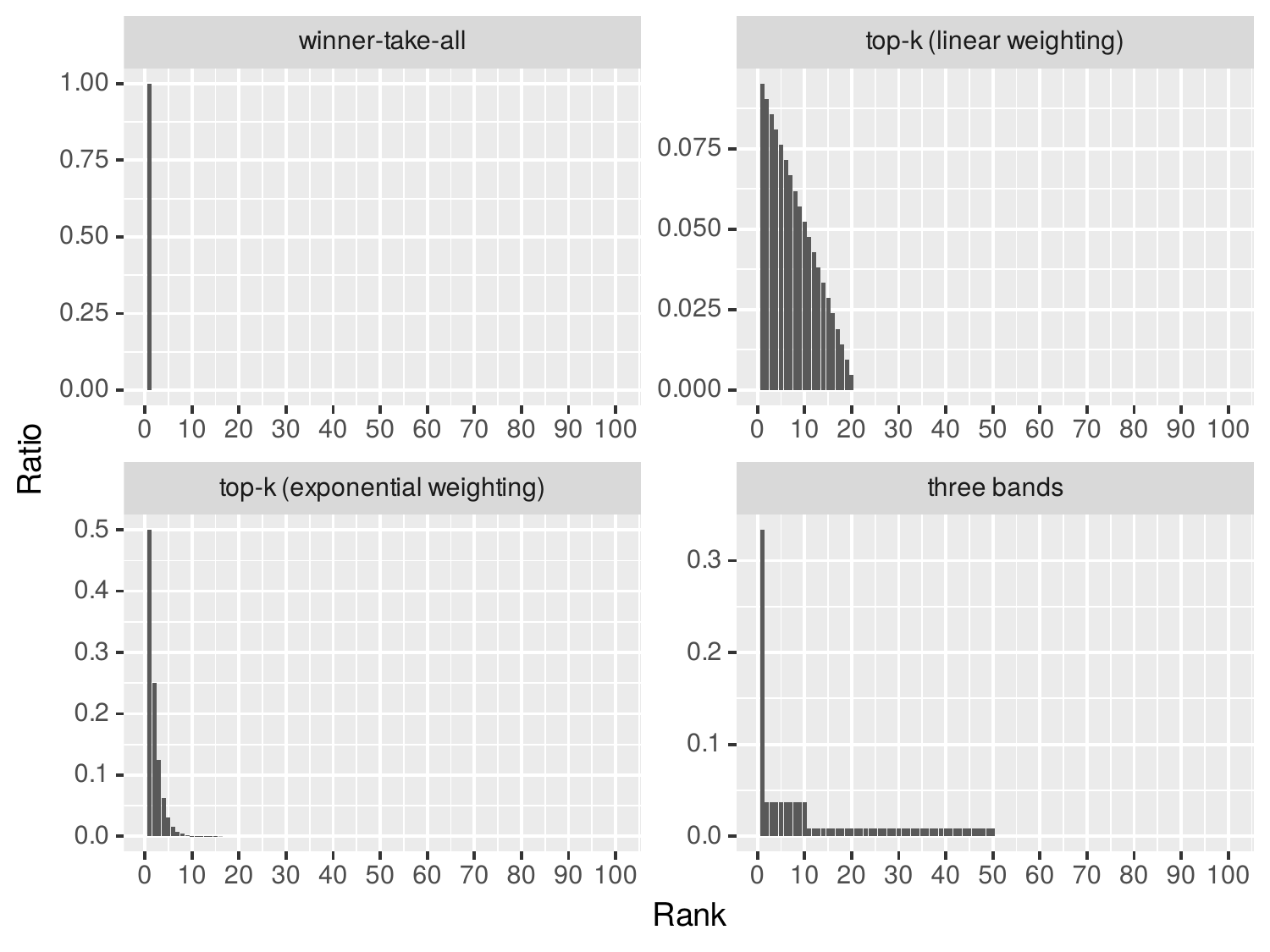}
    \caption{Comparison of prize distribution ($\nParticipants = 100$, and $k = 20$ for the top-$k$ mechanisms).}
    \label{fig:prize_distribution}
\end{figure}
As we cannot implement all the possible mechanisms, an operator needs to pick some promising ones. We compare the following four mechanisms in this paper.

\subsubsection{Winner-Take-All}
The winner-take-all mechanism is to give the top ranker everything, i.e. $\totalPrizes$, while others receive nothing.
\begin{equation}
    \Prize{j} = 
    \begin{cases}
        \totalPrizes, & \text{if } j = 1 \\
        0,            & \text{otherwise.} \\
    \end{cases}
    \label{eq:WTA}
\end{equation}

\subsubsection{Top-$k$ (Linear weighting)}
\label{sec:mechanism:top-k-linear}
The top-$k$ mechanism (linear weighting) is to share $\totalPrizes$ among top-$k$ rankers with linear weighting.
\begin{equation}
    \Prize{j} =
    \begin{cases}
        \totalPrizes \cdot \frac{\nParticipants - j + 1}{\nParticipants \cdot (\nParticipants + 1) / 2} , & \text{if } 1 \leq j \leq \nParticipants \\
        0,            & \text{otherwise.} \\
    \end{cases}
    \label{eq:top-k-linear}
\end{equation}
Note that $\nParticipants = 1$ corresponds to the winner-take-all mechanism.

\subsubsection{Top-$k$ (Exponential weighting)}
\label{sec:mechanism:top-k-exponential}
This mechanism is similar to the top-$k$ (linear weighting), but its weights are exponentially decreasing by ranks, meaning that the top ranker receives 50\% of $\totalPrizes$ and the runner-up receives 25\% of $\totalPrizes$ and so on.
\begin{equation}
    \Prize{j} =
    \begin{cases}
        \frac{\totalPrizes}{2^j}, & \text{if } 1 \leq j \leq \nParticipants \\
        0,            & \text{otherwise.} \\
    \end{cases}
    \label{eq:top-k-exponential}
\end{equation}
Note that although weights should be summed up to 1, the mechanism violates this. However, for sufficiently large $\nParticipants$, $\sum_{j=1}^{\nParticipants} 1/2^j \rightarrow 1$.

\subsubsection{Three bands}
This mechanism is just for comparison purposes. The top ranker receives $1/3$ of $\totalPrizes$, the second to tenth share another $1/3$ of $\totalPrizes$, the 11th to 50th share the remaining $1/3$ of $\totalPrizes$, and the remaining receive nothing. If the number of participants is below 50, then a game is terminated.
\begin{equation}
    \Prize{j} =
    \begin{cases}
        \frac{\totalPrizes}{3}, & \text{if } j = 1 \\
        \frac{\totalPrizes}{3 \cdot 9} = \frac{\totalPrizes}{27}, & \text{if } 2 \leq j \leq 10 \\
        \frac{\totalPrizes}{3 \cdot 40} = \frac{\totalPrizes}{120}, & \text{if } 11 \leq j \leq 50 \\
        0,            & \text{otherwise.} \\
    \end{cases}
    \label{eq:three-bands}
\end{equation}

\figurename~\ref{fig:prize_distribution} shows the proposed mechanisms' prize distribution versus ranks. In this figure, $\nParticipants = 100$ and $k = 20$ (20\% of participants will receive prizes for the top-$k$ mechanisms). As can be seen from this figure, the top ranker will receive all prizes in the winner-take-all game, whereas the top 20 (or 50) out of 100 will share a prize based on their ranks in the top-$k$ and three-band mechanisms.

\subsection{Cumulative Prospect Theory}
\label{sec:CPT}
CPT is one of the powerful theories that incorporates humans' decision-making bias under uncertainty into expected utility calculation. Let us give a simple example that violates the theory of simple expectation maximization. If we have two lotteries, (i) 50\% of a chance of winning \$200 and 50\% chance of losing \$100 and (ii) 100\% chance of winning \$49, and ask players which to play. Many players would choose the latter (i.e. a lottery that a player is sure to win) rather than the former even if the expected gain of the former one is larger (i.e. $-\$100 \cdot 0.5 + \$200 \cdot 0.5 = \$50 > \$49$). 
To capture such violations, Kahneman and Tversky presented a model of utility under risk called CPT, which extends traditional utility formulation to take into account humans' behavior on choices~\cite{Kahneman1979-pc, Tversky1992-vu}. The idea of CPT is to tweak the functions of valuation and probability so that they better explain people's choices under risk. In CPT, players evaluate the following equation rather than Eq.~\eqref{eq:EUT}. 
\begin{equation}
    \sum_{j = 1}^{\nParticipants} w\left(\frac{1}{\nParticipants}\right) v(\Profit{j}),
    \label{eq:CPT}
\end{equation}
where $v(\cdot)$ is a value function and $w(\cdot)$ are decision weights, respectively. In CPT, some characteristics are involved in designing $v(x)$, namely (i) reference dependence, (ii) loss aversion, and (iii) diminishing sensitivity. Reference dependence means that players care about gains and losses that are relative to some point rather than to the absolute amount of their wealth. Loss aversion means that players are more sensitive to losses than to gains, i.e. $|v(x)| < |v(-x)|$ for $x \geq 0$. Diminishing sensitivity means that players tend to be risk-averse in the region of gains but risk-loving in the region of losses. To capture these characteristics, the value function proposed by Kahneman and Tversky~\cite{Kahneman1979-pc, Tversky1992-vu}, and others (e.g~\cite{Prelec1998-fj}) is as follows:
\begin{equation}
    v(x) = \begin{cases}
             x^\alpha,               & \text{if } x \geq 0, \\
             -\lambda (-x)^{\alpha}, & \text{otherwise}.
           \end{cases}
    \label{eq:v_x}
\end{equation}

CPT also captures humans' behavior in they tend to overweight unlikely occurring events (i.e. the tails of the distribution). Specifically, $p_j$ is now transformed, through a function $w(\cdot)$, into $w(p_j)$. Tversky and Kahneman's weighting function, $w_\mathrm{TK}(p)$, is expressed as follows.
\begin{equation}
    w_\mathrm{TK}(p) = \frac{p^{\delta}}{\left(p^{\delta} + (1-p)^{\delta}\right)^{1/\delta}}.
    \label{eq:w_p_TK}
\end{equation}
Another example of probability weighting functions includes Prelec's as expressed as Eq.~\eqref{eq:w_p_Prelec}.
\begin{equation}
    w_\mathrm{Prelec}(p) = - \exp\{-\beta(- \ln p)^{\alpha}\}.
    \label{eq:w_p_Prelec}
\end{equation}

\begin{figure}[t]
    \centering
    \subfloat[Value function $v(x)$] {
        \includegraphics[width=.95\columnwidth]{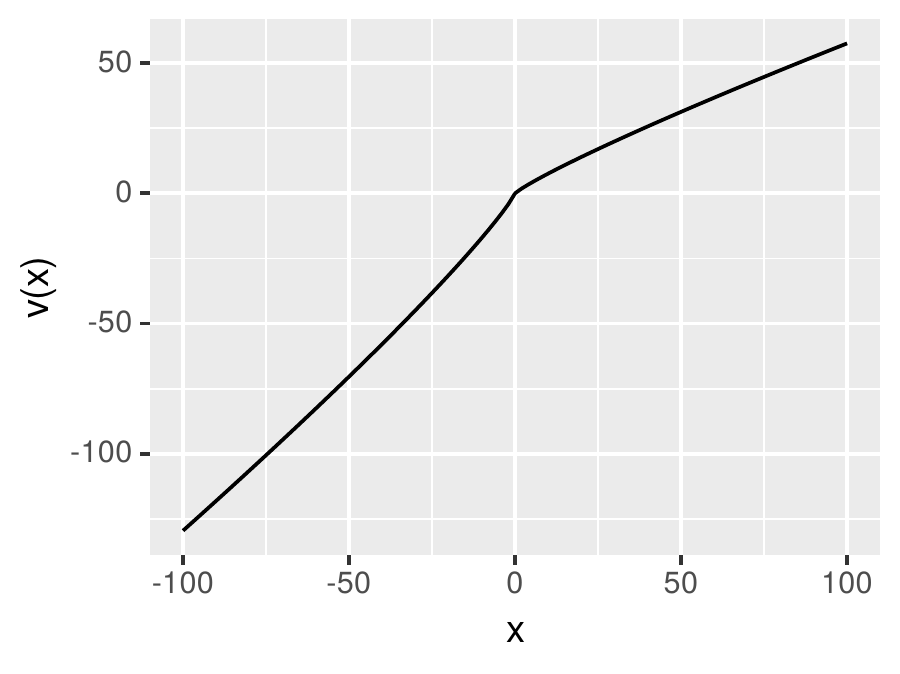}
        \label{fig:v_x}
    }\\
    \subfloat[Probability weighting functions $w(p)$] {
        \includegraphics[width=.95\columnwidth]{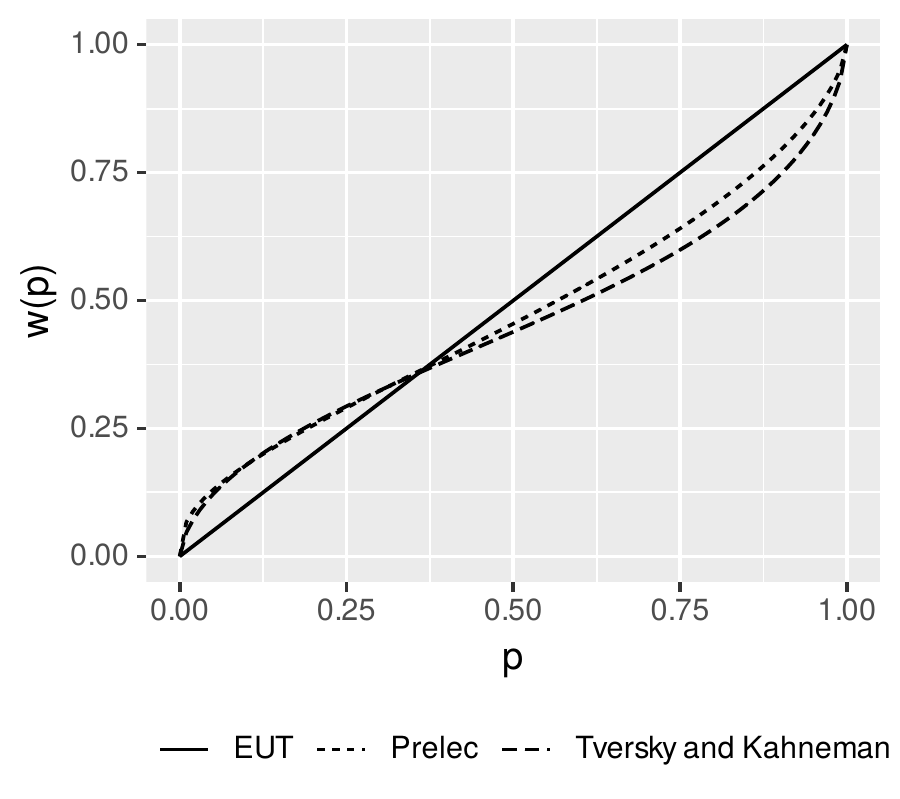}
        \label{fig:w_p}
    }
    \caption{The value function $v(x)$ and probability weighting functions $w(p)$ proposed by Tversky and Kahneman ($\alpha = 0.88, \lambda = 2.25$ for $v(x)$, and $\delta = 0.65$ for $w_\mathrm{TK}(p)$)~\cite{Tversky1992-vu} and by Prelec ($\alpha = 0.65, \beta = 1$ for their $w_\mathrm{Prelec}(p)$)~\cite{Prelec1998-fj}.}
\end{figure}
Figures~\subref*{fig:v_x} and \subref{fig:w_p} show how $v(x)$ and $w(p)$ transform the values of $x$ and $p$, respectively. As can be seen from \figurename~\subref*{fig:v_x}, a utility is concave when $x \geq 0$ and convex in $x < 0$ and $|v(x)| < |v(-x)|$ for $x \geq 0$. Similarly, $w(p)$ captures a human's tendency to overweight the possibilities of unlikely happening events, while no weighting is considered in the EUT.

\subsection{Comparison of Mechanisms}
With Eqs.~\eqref{eq:profit_participants}, \eqref{eq:CPT}, \eqref{eq:v_x} and \eqref{eq:w_p_TK} as well as $\Prize{j}$ of the aforementioned mechanisms (i.e. from Eq.~\eqref{eq:WTA} to Eq.~\eqref{eq:three-bands}) we can compare participant's utilities and choose the most promising mechanism for an operator. In this regard, we need to determine a value function, probability weighting function, and their parameters. 
The parameters of these functions (e.g. Eqs.~\eqref{eq:v_x} and \eqref{eq:w_p_TK}) are often determined via social experiments (e.g.~\cite{Tversky1992-vu, Prelec1998-fj, Gonzalez1999-jr, Bruhin2010-kh}). Subjects are recruited to join experiments, and the parameters are inferred from the experiments' results. Hence, when using the parameters obtained in this way, we need to make sure that the games in their experiments should be similar to those at hand. In this paper, we compare (i) Tversky and Kahneman's functions and (ii) Prelec's functions with the parameters inferred by their past social experiments (e.g.~\cite{Tversky1992-vu, Prelec1998-fj}). The reason for the choice of these two is that they have been well used to explain the behavior of lotteries in the past (e.g.~\cite{Barberis2012-dr}).

\section{Performance Evaluation}
\label{sec:performance_evaluation}
We compare the mechanisms described in the previous section, namely (i) winner-take-all, (ii) top-$k$ (linear weighting), (iii) top-$k$ (exponential weighting), and (iv) three-band mechanisms, in terms of participants' utility under the CPT assumption. 
We also clarify how $\CommissionFeeRatio$ and $\EntryFee$ affect the operator's profit and how profitable each mechanism is.
Regarding $\EntryFee$, we assume cardinal values (e.g. 1, 2) rather than actual currencies for generality. 
A positive utility means that a mechanism is attractive to participants, and they thus should join, and vice versa. Furthermore, the higher utility, the more attractive to participants. 
We determine the optimal $k$s for the two top-$k$ mechanisms (i.e. linear weighting and exponential weighting described in Sections~\ref{sec:mechanism:top-k-linear} and \ref{sec:mechanism:top-k-exponential}) and use them for the overall comparison.
Our codes are available at our GitHub repository.\footnote{ \url{https://github.com/kentaroh-toyoda/Design-of-Profitable-Crypto-Lottery-Mechanisms-with-Prospect-Theory}}

\subsection{Optimal $k$}
\begin{figure}
    \centering
    \subfloat[Effect of $k$ to utility.] {
        \includegraphics[width=\columnwidth]{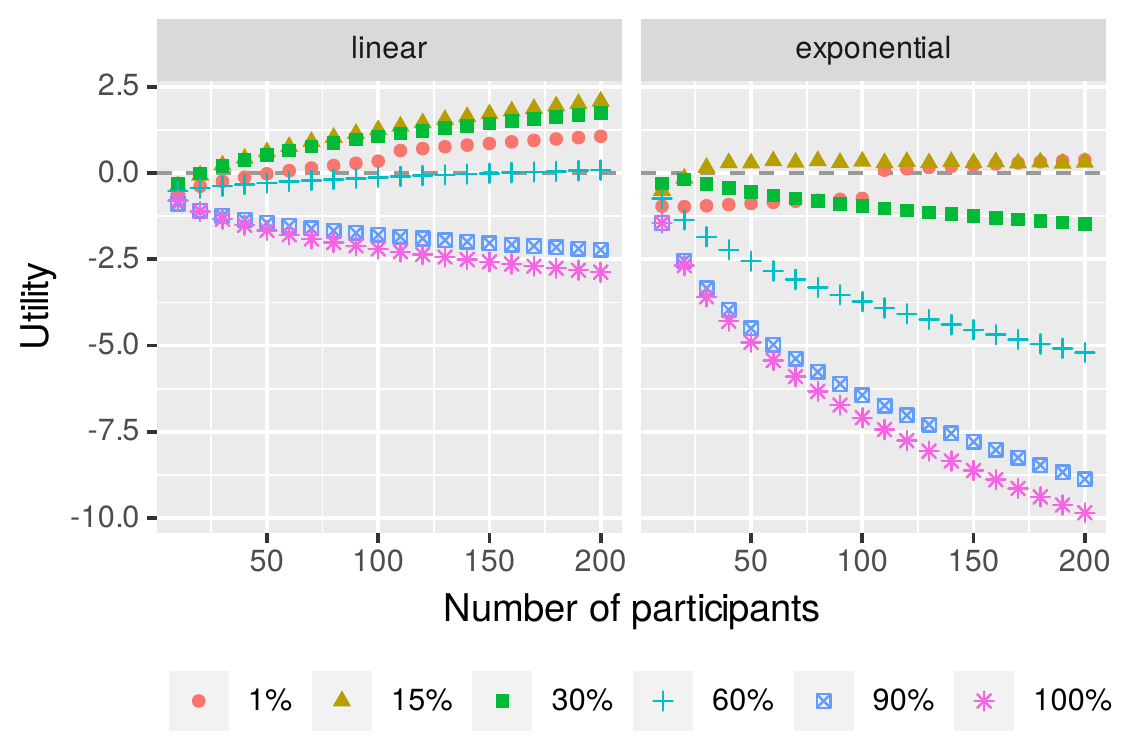}
        \label{fig:utility_k}
    }\\
    \subfloat[Average utility versus $k$. Optimal $k$ is 16\% for linear weighting and 6\% for exponential weighting.] {
        \includegraphics[width=\columnwidth]{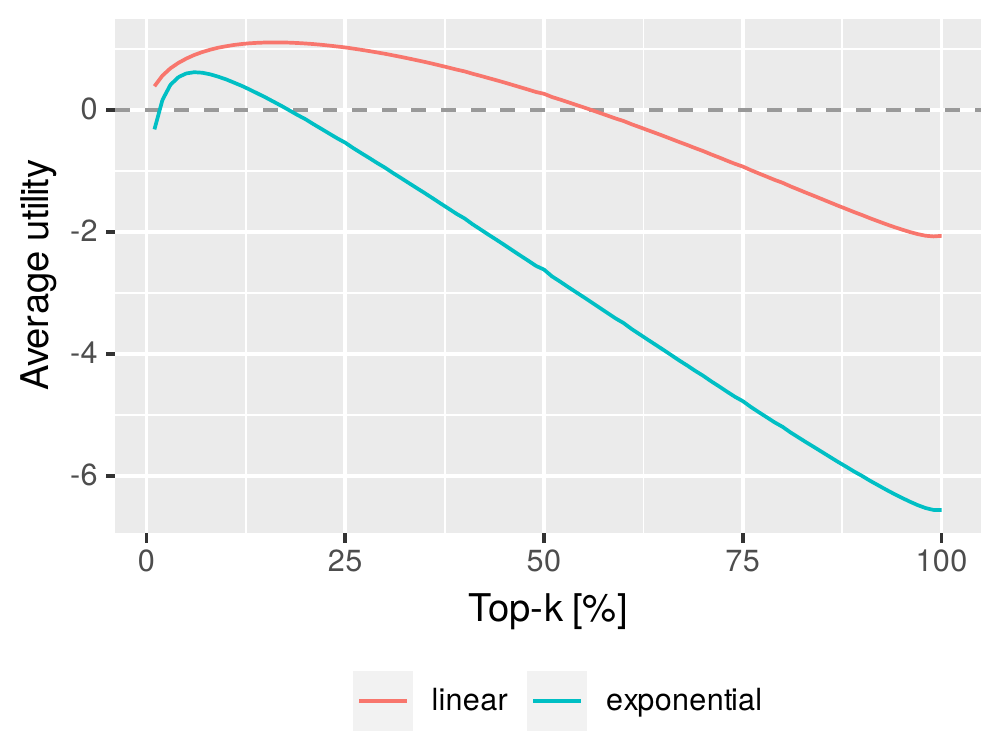}
        \label{fig:average_utility_k}
    }
    \caption{Effect of $k$ in linear and exponential weighting. We used Tversky and Kahneman's value function $v(x)$ with $\alpha = 0.88, \lambda = 2.25$ and probability weighting function $w_\mathrm{TK}(p)$ with $\delta = 0.65$~\cite{Tversky1992-vu}, $\CommissionFeeRatio = 10\%$, and $\EntryFee = 1$.}
\end{figure}
We determine the optimal $k$ for the two top-$k$ mechanisms (i.e. linear weighting and exponential weighting described in Sections~\ref{sec:mechanism:top-k-linear} and \ref{sec:mechanism:top-k-exponential}).
Here, we used Tversky and Kahneman's value function $v(x)$ with $\alpha = 0.88, \lambda = 2.25$ and probability weighting function $w_\mathrm{TK}(p)$ with $\delta = 0.65$~\cite{Tversky1992-vu}, $\CommissionFeeRatio = 10\%$, and $\EntryFee = 1$. $\nParticipants$ and $k$ were varied from 1 to 200 and from 1\% to 100\%, respectively. \figurename~\ref{fig:utility_k} shows the result. As can be seen from this figure, for both weighting methods, when $k$ is high (i.e. when most of the participants are expected to receive prizes but small amounts), it is less appealing to them. Similarly, when $k$ is extremely low (i.e. when only a few participants receive prizes), it is also less appealing. Hence, there should be somewhere that maximizes participants' utilities. We find such an optimal point by averaging utilities over the number of participants. \figurename~\ref{fig:average_utility_k} shows average utilities over $\nParticipants$ versus $k$. We can see from this figure that there are optimal points on both linear weighting ($k = 16\%$) and exponential weighting ($k = 6\%$). We use these $k$s for the following evaluation.

\subsection{Comparison of Mechanisms}
\begin{figure}
    \centering
    \includegraphics[width=\columnwidth]{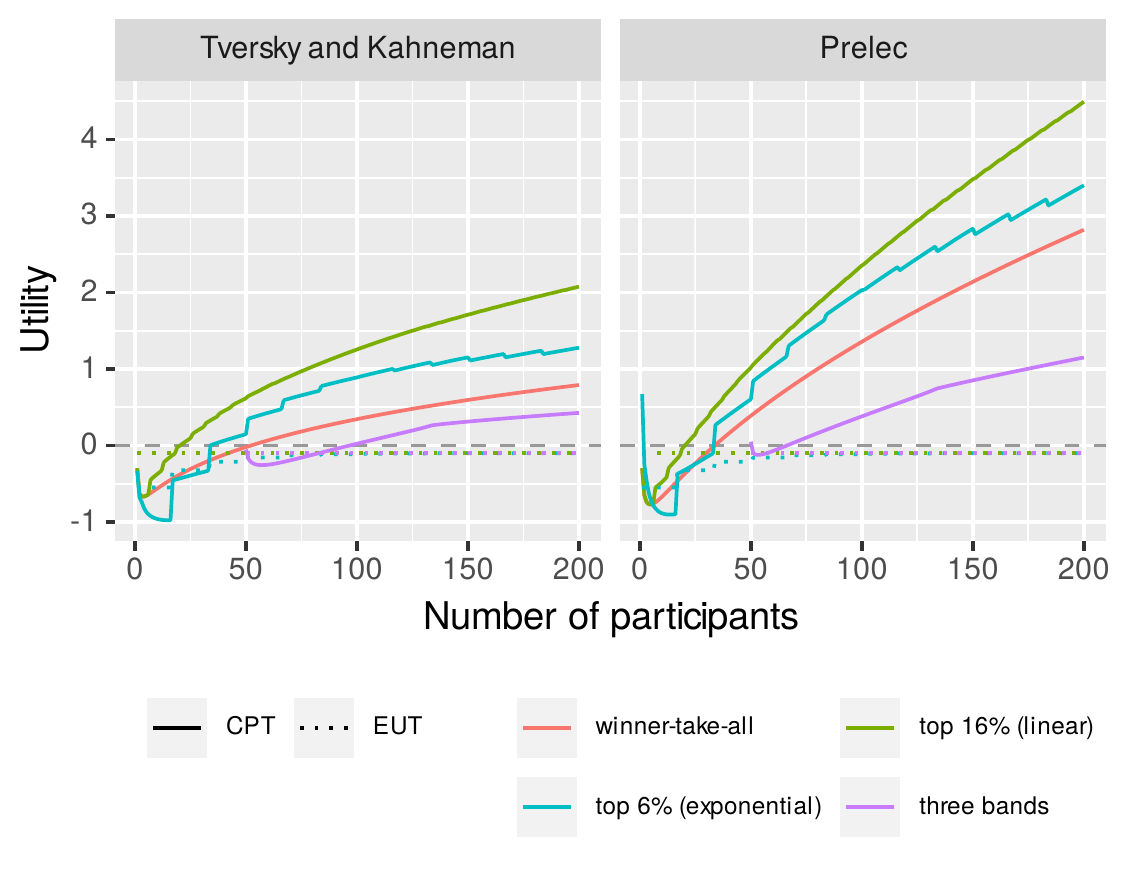}
    \caption{Utility comparison under the different assumptions.}
    \label{fig:utility}
\end{figure}
We clarify the relationships between utility and number of participants, $\nParticipants$, as well as different assumptions (i.e. value and probability weighting functions, CPT versus EUT).\footnote{Although we proved that the expected utility of risk-neutral participants is always negative, we plot it for reference purposes.} As mentioned in the previous section, we compare (i) Tversky and Kahneman's value function $v(x)$ with $\alpha = 0.88, \lambda = 2.25$ and probability weighting function $w_\mathrm{TK}(p)$ with $\delta = 0.65$~\cite{Tversky1992-vu} and (ii) Prelec's (the same $v(x)$ and $w_\mathrm{Prelec}(p)$ with $\alpha = 0.65, \beta = 1$)~\cite{Prelec1998-fj}, where these parameters were inferred by their past social experiments (e.g.~\cite{Tversky1992-vu, Prelec1998-fj}). We also assumed $\CommissionFeeRatio = 10\%$ and $\EntryFee = 1$ and varied $\nParticipants$ from 1 to 200. The optimal parameters of $k$ in the two top-$k$ mechanisms were used, namely $k = 16\%$ and $k = 6\%$ for linear and exponential weighting, respectively.

\figurename~\ref{fig:utility} shows the comparison of utilities by different mechanisms and assumptions. We can clearly see the difference in utilities under the CPT assumption. The top-16\% (linear weighting) is the most appealing mechanism for participants in the sense that (i) the minimum number of participants required to achieve a positive utility (i.e. 21 required for the top-16\% (linear weighting), 34 for the top-6\% (exponential weighting), 53 for the winner-take-all, and 97 for the three bands) and that (ii) it achieves higher utilities than the others. From this result, we can say that the winner-take-all and top-6\% (exponential weighting) are too risky and less appealing to participants and that the top-16\% (linear weighting) mechanism provides a good balance of risk and return. However, we can only say that it is the best among the four mechanisms tested; it is an open question to find the theoretically optimal mechanism. Likewise, although our findings are true against the value and probability weighting functions tested, they may not hold true for other functions. 

\subsection{Effect of $\EntryFee$}
\label{sec:effect_of_fee}
\begin{figure}
    \centering
    \includegraphics[width=\columnwidth]{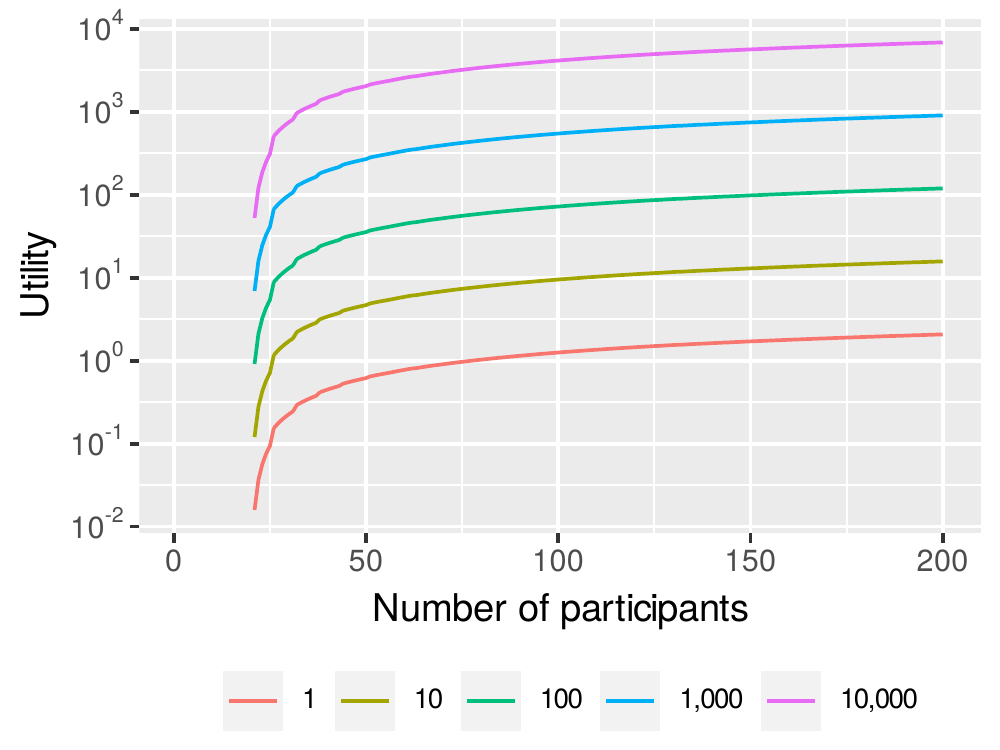}
    \caption{Effect of $\EntryFee$.}
    \label{fig:utility_fee}
\end{figure}
We clarify how $\EntryFee$ affects participants' willingness to join the game. \figurename~\ref{fig:utility_fee} shows the log-scale comparison of utilities when $\EntryFee$ is varied from 1 to 10,000. We used Tversky and Kahneman's functions with the same parameters and the top-16\% (linear weighting) mechanism for this evaluation. Regardless of the values of $\EntryFee$, utilities are positive when $\nParticipants \geq 21$. We can see from \figurename~\ref{fig:utility_fee} that the higher the entry fee, the more appealing to participants. In our game, as accumulated entry fees are distributed to top-$16\%$ after an operator takes $\CommissionFeeRatio$ of it, the more entry fee would be appreciated by participants. However, it is still an open question that this still holds true even if we replace $\EntryFee$ with a real fiat or cryptocurrency. For instance, is $\EntryFee = 1$ BTC, which is around US\$30,000 at the time of writing this paper, more appealing than $\EntryFee = 0.01$ BTC? Although this seems too risky, the proposed approach cannot explain this well yet. 

\subsection{Effect of $\CommissionFeeRatio$}
\label{sec:effect_of_r}
\begin{figure}
    \centering
    \includegraphics[width=\columnwidth]{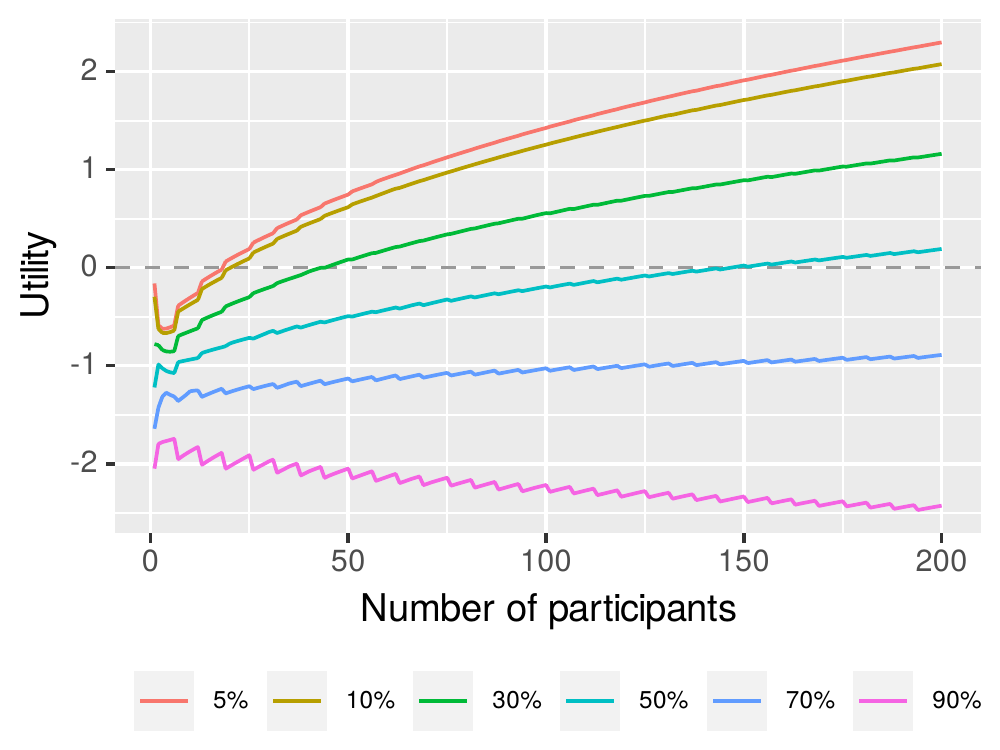}
    \caption{Effect of $\CommissionFeeRatio$.}
    \label{fig:utility_r}
\end{figure}
We then clarify how $\CommissionFeeRatio$ affects participants' utilities. \figurename~\ref{fig:utility_r} shows the utilities when $\CommissionFeeRatio$ is varied from 5\% to 90\%. As can be seen from this figure, the higher $\CommissionFeeRatio$, the less appealing to participants, which is quite understandable as participants will receive less amount of prize when $\CommissionFeeRatio$ is high. For extreme cases (e.g. $\CommissionFeeRatio = $ 90\%), utility keeps decreasing even if $\nParticipants$ increases. Furthermore, the higher $\CommissionFeeRatio$, the more participants are required for the utility to become positive. For instance, only 19 participants are required for $\CommissionFeeRatio = 5\%$, whereas 45 are required for $\CommissionFeeRatio = 30\%$. 

\subsection{Operator's Expected Profit}
\begin{figure}
    \centering
    \includegraphics[width=\columnwidth]{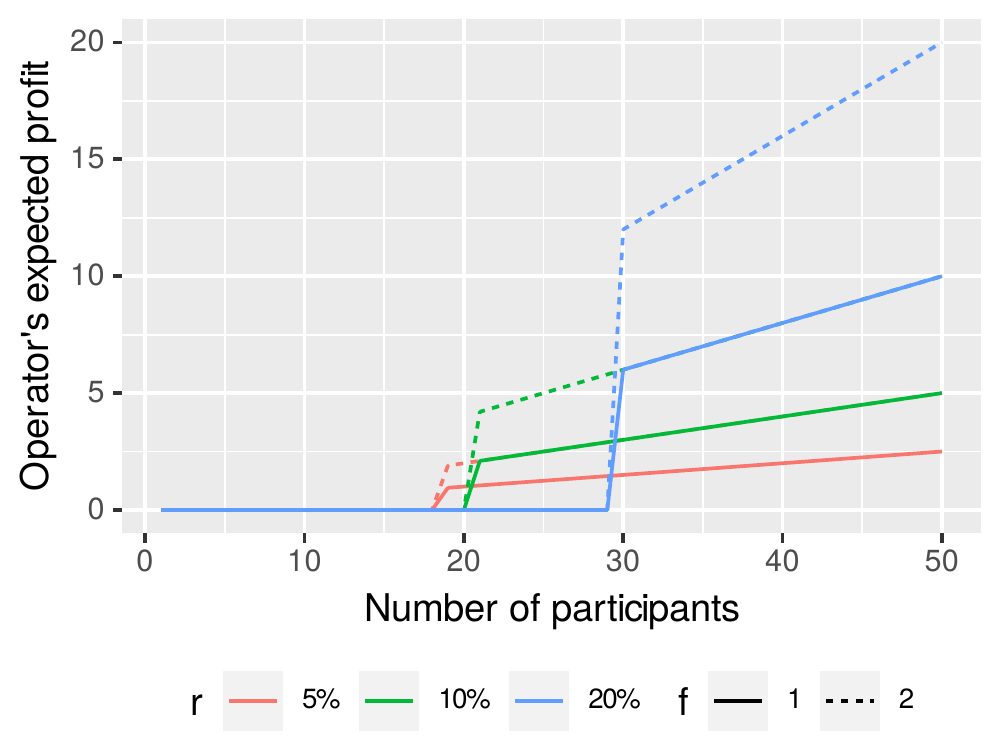}
    \caption{Expected Operator's profit.}
    \label{fig:operators_expected_profit}
\end{figure}
Finally, we compare the operator's expected profits when $\CommissionFeeRatio$ and $\EntryFee$ are varied. We use the same setting as Section~\ref{sec:effect_of_r} except that $\nParticipants$, $\CommissionFeeRatio$ and $\EntryFee$ are varied from 1 to 50, from 5\% to 20\% and from 1 to 2, respectively. \figurename~\ref{fig:operators_expected_profit} shows the operator's expected profits. Note that some results correspond to some regions. For instance, the results of $(\CommissionFeeRatio = 5\%, \EntryFee = 2)$ and $(\CommissionFeeRatio = 10\%, \EntryFee = 1)$ correspond when $\nParticipants \geq 21$. From the operator's perspective, $\CommissionFeeRatio$ should not be too small to increase profitability. However, the higher $\CommissionFeeRatio$, the more participants should join for a game to be appealing. In this regard, one of the operator's strategies is to infer the expected number of participants somehow and to set $\CommissionFeeRatio$ accordingly. For instance, if 30 participants may join a game, then $\CommissionFeeRatio$ should be less than 20\% as this is the maximum $\CommissionFeeRatio$.

\subsection{Discussion}
The contributions of our method are that it provides a reasonable way of designing crypto lottery mechanisms and comparing them in terms of utilities as well as profits. Besides, the method discussed is generic. We believe that it can be applied to other use cases such as the design and analysis of optimal blockchain mining reward and token-based applications such as crowdsourcing and data mining platforms.

On the flip side, we are aware of the following limitations which remain open questions.
\begin{itemize}
    \item Obtained results and conclusions may be biased by a chosen value function, probability weighting function, and their parameters. It is necessary to check if the conditions of experiments where such parameters are derived are similar to the problems at hand.
    \item An operator must manually come up with possible mechanisms, and thus the truly optimal mechanism may be missed.
    \item The model of CPT used in this paper does not explain well how fiat or cryptocurrency $\EntryFee$ affects participants' willingness to join. Is $\EntryFee = 1$ BTC more appealing than $\EntryFee = 0.01$ BTC even if the former is too risky?
    \item This analysis only focuses on a single game. However, it may be necessary to analyze mechanisms and utility when a game is repeated as in~\cite{Barberis2012-dr}.
    \item More factors may need to be considered in modeling participants' behavior. For instance, there is a report that risk attitudes differ by country and depend not only on economic conditions but also on cultural factors~\cite{Rieger2011-zr}.
\end{itemize}

\section{Conclusions}
\label{sec:conclusions}
We have proposed a method of designing the mechanism of lottery games with an example of a crypto-based lottery game.
The key idea is to incorporate behavioral economics into mechanism design to better predict participants' willingness to join a game and the operator's profitability based on utility analysis. In particular, we leveraged CPT to model participants' behavior. We proposed four mechanisms for the game and thoroughly evaluated them in terms of utility and profit by varying parameters.
Our evaluation suggests that the top-$k$ (linear weighting), which distributes prizes to top-$k$ participants of a game and the amount of prizes linearly declines with their ranks, is the best mechanism among the four mechanisms. We have also clarified the relationships between utility and the number of participants under some assumptions (e.g. how the parameters of how much an operator takes an entry fee affect utilities and the minimum number of participants required to make a game appealing).

Our method has some contributions in designing a crypto-based lottery game, however, there is some room for improvement. For instance, ours do not explain well that utilities keep increasing even if a game is too risky to join when an entry fee is too high. We will tackle these issues and provide a better mechanism design method.

\section*{Acknowledgment}
We thank Xuan-The Tran, Minh Sang Nguyen, and Hung Tien Dinh from Crypto Sloth for their valuable feedback.

\bibliographystyle{IEEEtran}
\bibliography{ref}

\end{document}